\documentclass{article}

\usepackage{arxiv}

\usepackage[utf8]{inputenc} 
\usepackage[T1]{fontenc}    
\usepackage{booktabs}       
\usepackage{amsfonts}       
\usepackage{nicefrac}       
\usepackage{microtype}      
\usepackage{lipsum}
\usepackage{graphicx}
\usepackage{amsmath}
\usepackage{amsfonts}
\usepackage{graphicx}      
\usepackage{subcaption}
\usepackage{float}

\newtheorem{proof}{Proof}
\newtheorem{prop}{Proposition}
\graphicspath{ {./images/} }

\title{Output adaptive control for linear systems under parametric uncertainties with finite-time matching input harmonic disturbance rejection}

\author{
 Dmitrii Dobriborsci \\
  Faculty of Control Systems and Robotics\\
  ITMO University\\
  St. Petersburg, Russia \\
  \texttt{dmitrii.dobriborsci@itmo.ru} \\
   \And
 Sergey Kolyubin \\
   Faculty of Control Systems and Robotics\\
  ITMO University\\
  St. Petersburg, Russia \\
  \texttt{s.kolyubin@itmo.ru} \\
  \And
 Alexey Bobtsov \\
   Faculty of Control Systems and Robotics\\
  ITMO University\\
  St. Petersburg, Russia \\
  \texttt{bobtsov@mail.ru} \\
}

\begin{document}
\maketitle
\begin{abstract}
We consider the task of motion control for non-prehensile manipulation using parallel kinematics mechatronic setup, in particular, stabilization of a ball on a plate under unmeasured external harmonic disturbances. System parameters are assumed to be unknown, and only a ball position is measurable with a resistive touch sensor. To solve the task we propose a novel passivity-based output control algorithm that can be implemented for unstable linearized systems of an arbitrary relative degree. In contrast to previous works, we describe a new way to parametrize harmonic signal generators and an estimation algorithm with finite-time convergence. This scheme enables fast disturbance cancellation under control signal magnitude constraints. 
\end{abstract}


\section{Introduction}
The development of control algorithms for robot manipulators with a parallel kinematic scheme is an pressing challenge in the tasks of dynamic manipulation. Such systems have several advantages compared to manipulators with a serial scheme: kinematic chains are closed, which leads to robustness, as well as high accuracy of the positioning of the mechanism as a whole. Movable parallel parts reduce the load on the drive, which improves the dynamics and accuracy of the system \cite{Lynch199964}. Similar systems are used in flight simulators, in simulators for car drivers, in the production process. Parallel kinematics robots are also widely used in biomechatronics and rehabilitation of the neck, knee joints and foot joints. For example, in rehabilitation, various variations of the Gaugh-Stewart platform with 6 degrees of freedom are used. Some scientists have proposed to use simplified models with two or three degrees of freedom. The robotic setup developed during the study in the general case is a parallel manipulator robot with two degrees of freedom. In tasks of nonprehensile manipulation, it becomes possible to use such techniques as pushing, throwing, hitting, tilting. As a result, the scope of application of robotic manipulators in the industrial world expands \cite{Dobriborsci2017}. The developed robotic setup allows to conduct experiments such as positioning an object predefined coordinates, moving along a given path, identifying a dynamic model, as well as conducting experimental testing of modern control algorithms and information processing

This research is devoted to the output-feedback control of  linear parametrically uncertain plants under unmeasured matching input harmonic disturbances. The paper presents a novel switching control algorithm that combines a passification-based output controller with a finite-time harmonic disturbance parameters estimation algorithm that guarantees convergence of the disturbance parameters. Thus, the overall closed-loop system performance can be improved. 

It is assumed that the frequency of the harmonic signal is unknown. In the majority of works devoted to the synthesis of frequency identification algorithms the possibility of increasing the rate of convergence is not discussed, which can also be attributed to the open problems of control theory.

In the present work, the method of a 'consecutive compensator' \cite{6044373} is used as a basic control approach. This method had been proven efficient in a number of parametrically uncertain robotic and mechatronic systems control applications, see e.g.  \cite{Dobriborsci2018655, Dobriborsci201862, Dobriborsci2018533}. Moreover, the approach guarantees global convergence, can be applied to systems with an arbitrary relative degree, has a simple structure and is easy to configure. We implement this output-feedback controller to stabilize an unstable parametrically uncertain plant and further use its output for the input disturbance model identification.

The problem of cancelling external harmonic disturbances acting on unstable parametrically and structurally uncertain plants by identifying disturbance's internal model parameters as well as ways of increasing the rate of parametric convergence were studied in the authors' previously published works, e.g. \cite{Pyrkin2015,Pyrkin20151}. 

Combining these two results for output adaptive control with simultaneous disturbance cancellation was suggested in the authors' preceding works and recently applied to the Ball-and-Plate mechatronic system control, see \cite{8795674}. The suggested switching scheme assumes substitution of the obtained frequency estimates as output adaptive stabilizing controller, see \cite{Pyrkin201412110} for details. Such an approach lets us assign switching intervals given a certain decay rate of estimates' discrepancies and therefore avoid undesired spikes and oscillations in closed-loop transients.

However, finite-time estimation algorithms and controllers based on 'consecutive compensator' approach were never combined before to solve the output control problem under parametrical and signal disturbances. Such a fusion is quite promising, it provides convergence of the disturbance estimates with a finite amount of time. A number of papers are devoted to finite-time estimators. In \cite{og2019}, \cite{Gerasimov2018886} an estimator design which provides finite-time convergence is proposed. In \cite{Wang20192963} an adaptive estimator of constant parameters without  the hypothesis that regressor is Persistently Excited (PE).

The approach described in this paper is close to the monitoring function method presented in \cite{RouxOliveira20172684} for the problem of the output adaptive tracking control. The difference is that we use output adaptive robust controller with high-gain observer, which is different from sliding mode control.
 
The paper is structured as follows. After a short introduction, the problem statement together with important assumptions are given. Section~3 describes the output stabilizing controller based on 'consecutive compensator' method. In Section~4 we present the finite-time disturbance frequency estimation algorithm for input harmonic disturbance parameters, then in Section~\ref{sw} a switching scheme that enables using these estimates for feedback controller adjustments. Finally, Section~5 is devoted to the case study of Ball-and-Beam robotic platform control using the presented technique. As an example we consider the problem of ball stabilization on a square platform. The task is complicated by the presence of harmonic disturbances in the system. Obtained results illustrate the overall improved performance of the system. 

\section{Problem statement}

Consider the linear SISO plant
\begin{equation} 
\label{plantIO} 
a(p)y(t)=b(p)[u(t)+\delta (t)], 
\end{equation} 
where $p=\frac{d}{dt} $ is the differentiation operator, $u(t)$ and $y(t)$ are input and output signals respectively, coefficients of the polynomials $a(p)=p^{n} +a_{n-1} p^{n-1}  +...+a_{0} $ and $b(p)=b_{m} p^{m} +b_{m-1} p^{m-1}  +...+b_{0} $ are unknown, and 
\begin{align}
	\delta(t)= \bar{A}\sin(\omega t + \bar{\phi}) 
	\label{distr1}
\end{align}
is the input harmonic disturbance with the unknown amplitude $\bar{A}$, phase shift $\bar{\phi}$, and  frequency 

${0 < \omega_{min} <\omega < \omega_{max} < \infty}$.



The control goal is to guarantee closed-loop system output stability
\begin{equation} 
\label{purpose_of_control} 
\lim_{t\to\infty} y(t)=0 
\end{equation} 
under the following assumptions:
\begin{enumerate}
\item $b(p)$ is a Hurwitz polynomial;
\item only the relative degree of the system $\rho = n - m$ is known, while degrees of the polynomials $a(p)$ and $b(p)$ are unknown.
\item The lower bound $\omega_{min}$ of frequency $\omega$ is known.
\end{enumerate}

\section{Output controller design}

Let us consider the output adaptive controller with modification for the input harmonic disturbance rejection  introduced in \cite{Bobtsov2012}
\begin{align} 
\label{control} 
&u(t)=-k\frac{\alpha (p)(p+1)^{2}}{(p^2 + \omega^2)}\xi _{1} (t), \\
\label{consecutive_compensator_xi}
&\left\{\!\!\begin{array}{l}
\dot\xi_1=\sigma\xi_2,\\
\dot\xi_2=\sigma\xi_3,\\
\dots\\
\dot\xi_{\rho_{m}-1}=\sigma\left(-k_1\xi_1-\dotsc-k_{\rho-1}\xi_{\rho-1}+k_1 y\right),
\end{array}\right.
\end{align} 
where $\alpha (p)$ is a Hurwitz polynomial of $(\rho-1)$ degree, constant coefficient $k > 0$ is chosen such way that transfer function 
\[
{H(p)=\frac{\alpha(p)b(p)(p+1)^{2}}{a(p)(p^2 + \omega^2) + k \alpha(p)b(p)(p+1)^{2}}}
\]
is SPR, while $\sigma > k$ and parameters $k_i$ are calculated for the system \eqref{consecutive_compensator_xi} to be asymptotically stable for $y(t) = 0$. 

Substitution of \eqref{control} into \eqref{plantIO} yields to the closed-loop system description 
\begin{align} 
\label{closedloop2} 
y(t)&=\frac{k b(p)\alpha (p)(p+1)^{2} }{a(p)(p^2 + \omega^2)+kb(p)\alpha (p)(p+1)^{2}} \varepsilon (t) \nonumber \\
&+\frac{b(p)(p^2 + \omega^2)}{a(p)(p^2 + \omega^2)+kb(p)\alpha (p)(p+1)^{2}} \delta (t),
\end{align}
where $\varepsilon (t)=y(t)-\xi _{1} (t)$.

We can further rewrite \eqref{closedloop2} as
\begin{align} 
\label{closedloop3} 
y(t)&=\frac{k b(p) \alpha (p)(p+1)^{2} }{a(p)(p^2 + \omega^2)+kb(p)\alpha (p)(p+1)^{2}}\nonumber \\
&\times [\varepsilon (t)+w(t)], 
\end{align} 
where a signal $w(t)=\frac{(p^2 + \omega^2)}{k \alpha (p)(p+1)^{2} } \delta (t)$.

Let us write the Laplace representations of the disturbance signal \[
\Psi(s)=L\{ \delta (t)\}=\frac{\delta_0}{s}+\frac{\mu \omega + \nu s}{s^2 + \omega^2},
\]

where $s$ is the complex variable.

Then, we can show that 
\begin{align*}
L\{w(t)\} = \frac{s(s^2 + \omega^2)}{k \alpha (s)(s+1)^{3} } \left[\frac{\delta_0}{s}+\frac{\mu \omega + \nu s}{s^2 + \omega^2} \right],
\end{align*} and therefore $w(t)$ is vanishing with time.

Now we can derive the state-space form of \eqref{closedloop3}
\begin{align} 
\label{closedloop4} 
\dot{x}&=Ax+kb(\varepsilon +w), \\
\label{closedloop4_y} 
y&=c^T x, 
\end{align} 
where $x\in R^{n} $ is the state vector of the model \eqref{closedloop4}; $A$, $b$, and $c$ are the corresponding matrices. 

Since by design ${\gamma(p) = a(p)p(p^2 + \omega^2)+kb(p)\alpha (p)(p+1)^{3}}$ is a Hurwitz polynomial, and according to the well-known Kalman-Yakubovich-Popov lemma one can take the positive symmetric matrix $P$, satisfying the following matrix equalities
\begin{equation} 
\label{KYP} 
A^T P+PA=-Q_{1} , \qquad  Pb=c,      
\end{equation} 
where $Q_{1} =Q_{1}^T $ is some positive definite matrix.

Let us next rewrite model \eqref{control}, \eqref{consecutive_compensator_xi} in the state-space form as well
\begin{align}
\label{xi_dot}
\dot \xi(t)&=\sigma(\Gamma\xi(t)+dy(t)),\\
\label{y_hat}
\hat y(t)&=h^{T}\xi(t),
\end{align}
where $\Gamma=\begin{bmatrix}
 0 & 1 & 0 & \dotsc & 0 \\
 0 & 0 & 1 & \dotsc & 0 \\
 0 & 0 & 0 & \dotsc & 0 \\
 \vdots & \vdots & \vdots & \ddots & \vdots \\
 -k_1 & -k_2 & -k_3 & \dotsc & -k_{\rho-1} \\
\end{bmatrix}$,  
$d=\begin{bmatrix}
 0 \\
 0 \\
 0  \\
 \vdots \\
 k_1 \\
\end{bmatrix}$, and $h^{T}=\begin{bmatrix}
 1 & 0 & 0 & \dotsc 0
\end{bmatrix}$ .

Consider the vector 
\begin{equation}
\eta(t)=hy(t)-\xi(t),
\end{equation}
then by force of vector $h$ structure the error $\varepsilon(t)$ will become
\begin{align}
\varepsilon(t)&=y(t)-\hat y(t)=h^{T}hy(t)-h^{T}\xi(t)\nonumber\\
&=h^{T}(hy(t)-\xi(t))=h^{T}\eta(t).
\end{align}

For the derivative of $\eta(t)$ we obtain
\begin{align}
\dot\eta(t)&=h\dot y(t)-\sigma(\Gamma(hy(t)-\eta(t))+dy(t))\nonumber\\
&=h\dot y(t)+\sigma\Gamma\eta(t)-\sigma(d+\Gamma h)y(t).
\end{align}

Since $d=-\Gamma h$ (can be checked by substitution), then we have a complete state-space description of the closed-loop system
\begin{align}
\label{closedloop5} 
\dot{x}(t)&=Ax(t)+kb(\varepsilon(t) +w(t)), &\quad y(t)=c^T x(t),\\
\label{eta_system}
\dot\eta(t)&=h\dot y(t)+\sigma\Gamma\eta(t), &\quad \varepsilon(t)=h^T\eta(t),
\end{align}
where matrix $\Gamma$ is Hurwitz by force of calculated parameters $k_i$, and 
\begin{equation}
\label{lyap_gamma}
\Gamma^TN+N\Gamma=-Q_2,
\end{equation}
where $N=N^T>0$, $Q_2=Q_2^T>0$.

\begin{prop}
\label{main_proposition}
The output feedback controller \eqref{control}, \eqref{consecutive_compensator_xi} applied to the plant \eqref{plantIO} guarantees achievement of the control goal \eqref{purpose_of_control} for the output variable $y(t)$.  
\end{prop}

\begin{proof}
Following the ideas of \cite{BOBTSOV2013408}, choose the Lyapunov function
\begin{equation} 
\label{lyapunov} 
V=x^T Px+\eta ^T N\eta . 
\end{equation} 
Differentiation of \eqref{lyapunov} yields
\begin{align}
\label{lyapunov_dot}
\dot{V}&=x^{T} (A^{T} P+PA)x+2kx^{T} Pbh^{T} \eta +2kx^{T} Pbw \nonumber\\
&+\eta ^{T} \sigma (\Gamma ^{T} N+N\Gamma )\eta +2\eta ^{T} Nhc^{T} Ax\nonumber\\
&+2k\eta ^{T} Nhc^{T} bw+2k\eta ^{T} Nhc^{T} bh^{T} \eta .    
\end{align}
Consider inequalities  
\begin{align}
\label{inequalities}
2kx^{T} Pbh^{T} \eta &\le k^{-1} \, x^{T} Pbb^{T} Px+k^{3} \eta ^{T} hh^{T} \eta , \nonumber\\
2kx^{T} Pbw&\le k^{-1} \, x^{T} Pbb^{T} Px+k^{3} \, w^{2} ,  \nonumber\\
2k\eta ^{T} Nhc^{T} bh^{T} \eta &\le k\eta ^{T} Nhc^{T} bb^{T} ch^{T} N\eta +k\eta ^{T} hh^{T} \eta ,   \nonumber\\
2\eta ^{T} Nhc^{T} Ax&\le k\eta ^{T} Nhc^{T} AA^{T} ch^{T} N\eta +k^{-1} \, x^{T} x, \nonumber\\
2k\eta ^{T} Nhc^{T} bw&\le k\eta ^{T} Nhc^{T} bb^{T} ch^{T} N\eta +k\, w^{2}.
\end{align} 
Thus
\begin{align}
\label{lyapunov_dot2}
\dot{V}&\le -x^{T} Q_{1} x-\sigma \eta ^{T} Q_{2} \eta +k^{-1} \, x^{T} Pbb^{T} Px+k^{3} \eta ^{T} hh^{T} \eta \nonumber\\ 
&+k^{-1} \, x^{T} Pbb^{T} Px+k^{3} \, w^{2} +k\eta ^{T} Nhc^{T} bb^{T} ch^{T} N\eta   \nonumber\\
&+k\eta ^{T} hh^{T}\eta+k\eta ^{T} Nhc^{T} AA^{T} ch^{T} N\eta \nonumber\\ 
&+k^{-1} \, x^{T} x+k\eta ^{T} Nhc^{T} bb^{T} ch^{T} N\eta +k\, w^{2} .  
\end{align}

Let numbers $k>0$ and $\sigma >0$ be such that
\begin{align}
&-Q_{1} +k^{-1} Pbb^{T} P+k^{-1}Pbb^{T}P+k^{-1}I &\le -Q'<0,\nonumber\\ 
&-\sigma Q_{2}+(k+k^{3}) hh^{T}+kNhc^{T} bb^{T} ch^{T} N& \nonumber\\ 
&+kNhc^{T} AA^{T} ch^{T} N +k Nhc^{T} bb^{T} ch^{T} N &\le -Q'' <0,\nonumber
\end{align} 
then for the derivative of \eqref{lyapunov} we have

\begin{align}
\label{lyapunov_dot3}
\dot{V}\le -x^{T} Q'x-\eta ^{T} Q''\eta +(k^{3} +k)w^{2} .
\end{align} 

Hence it is easy to obtain the inequality
\begin{equation} 
\label{lyapunov_dot4} 
\dot{V}\le -\lambda V+(k^{3} +k)w^{2} , 
\end{equation} 
where $\lambda >0$.

Since $w(t)$ is a decaying term, condition \eqref{lyapunov_dot4} guarantees stability of the closed-loop system.

\end{proof}

\section{Finite-time Disturbance Frequency Estimation}
\label{drem}

Here we introduce the algorithm for finite-time input harmonic disturbance parameters estimation.

At first, we parametrize the disturbance model.

Since the considered closed-loop system is linear and stable, the output variable $y(t)$ (when the transient time has elapsed) is tracking the external disturbance, i.e. $y(t)= A\sin(\omega t + \phi)$.

Consider two auxiliary signals 
\begin{equation}
	\label{y1}
	y_1(t) = y(t-\tau),\\
\end{equation}
\begin{equation}
	\label{y2}
	y_2(t) = y(t-2\tau),
\end{equation}
where $\tau \in \mathbb{R}_+$ are chosen values of the delay duration.

Rewrite \eqref{y1} and \eqref{y2}
\begin{align}
	\label{par1}
	y_1(t) &= A\sin(\omega t +\phi) \cos \omega \tau - A\cos(\omega t +\phi) \sin\omega \tau,\\
	y_2(t) &= A\sin(\omega t +\phi) \cos 2\omega \tau - A\cos(\omega t +\phi) \sin 2\omega \tau.
\end{align}

Multiplying $y_1(t)$ by $\sin 2\omega\tau$ and $y_2(t)$ by $\sin\omega\tau$ and applying double angle formulas, we get
\begin{align}
	\label{par2}
	&y_1(t)\sin 2\omega\tau -y_2(t)\sin\omega\tau  = A\sin(\omega t +\phi) \cos \omega \tau \sin 2\omega\tau \nonumber\\
	& - A\cos(\omega t +\phi) \sin\omega \tau \sin 2\omega\tau \nonumber\\
	&- A\sin(\omega t +\phi) \cos 2\omega \tau \sin\omega \tau \nonumber\\ 
	& + A\cos(\omega t +\phi) \sin 2\omega \tau \sin\omega\tau = \nonumber \\
	&= 2 A\sin(\omega t +\phi) \cos^2 \omega \tau \sin \omega\tau \nonumber \\
	&- A\sin(\omega t +\phi) (2\cos^2 \omega \tau -1)\sin\omega \tau = \nonumber \\
	&=  A\sin(\omega t +\phi)\sin\omega \tau = y(t)\sin\omega \tau
\end{align}

Dividing \eqref{par2} by $\sin\omega \tau$, we get
\begin{equation}
	2 y_1(t)\cos\omega\tau = y(t)+y_2(t).
\end{equation}

Now, we can derive the exact model (without assumption on exponential decaying terms due to unknown initial conditions) of the harmonic signal generator in the linear regression form
\begin{equation}
    z(t) = \varphi(t) \theta,
    \label{new_form}
\end{equation}
where $z(t) = \frac{1}{2}\left(y(t)+y_2(t)\right)$, $\varphi(t) = y_1(t)$, and $\theta=\cos{\omega \tau}$.
Since the signals $y(t)$, $y_1(t)$, $y_2(t)$ are measurable, it is easy to obtain the disturbance frequency.

To estimate unknown parameters of the disturbance model \eqref{new_form} we use the following algorithm
\begin{align}
\label{est}
   \dot{\hat\theta}(t) &= K \varphi(t) (z(t) - \varphi(t)\hat{\theta}(t)), \\
  \label{est_finite}
   \hat{\theta}_F(t) &= \frac{1}{1-w(t)} \left[ \hat{\theta}(t)-w(t)\hat{\theta}_0 \right], \\
   \label{est_w}
   \dot{w}(t) &= -K \varphi^2(t)w(t),
\end{align}
where $K>0$, $w(0) = 1$, and $\hat{\theta}_0$ is an initial guess on the disturbance parameters values. 

It can be shown that $\hat\theta_F(t)$ converges to the real value of $\theta$ for finite time, which can be reduced by adjusting the gain $K$. The only issue with selecting very high values for $K$ is that the presented scheme becomes very sensitive to measurements noise. But in any case we need to wait some small amount of time before $w(t) < 1$ such that \eqref{est_finite} does not give division by zero. Detailed description and proof for \eqref{est}--\eqref{est_w} could found in \cite{og2019}.

\section{Switching scheme}
\label{sw}

The last step in the proposed scheme is to set up a criterion, which would implement disturbance parameters' estimates substitution to the nominal controller \eqref{control} $\omega(t_i) = \bar{\omega}(t_i)$.

As it was outlined above, in contrast to \cite{8795674} we perform a single switching, i.e. the substitution of the disturbance frequency estimates will be performed at the moment of time when its already converged. 
The switching scheme can be analytically described by the relations below and applied by using a trigger scheme:
\begin{equation}
\label{cases}
 \begin{cases}
 \bar\omega(t) = \omega_{min}, &\text{where $t < T$},\\
 \bar\omega(t) = \frac{\arccos\hat{\theta}_F(T)}{\tau}, &\text{where $t \geq T$},
 \end{cases}
\end{equation}
where $\hat{\theta}_F(T)$ is obtained from \eqref{est}--\eqref{est_w}.

A method reported above is quite similar to widely-used dwell-time switching logic and allows to avoid undesired jumps and oscillation in transients that can lead to loosing closed-loop system stability.

\begin{figure*}[ht]
    \centering
    \begin{subfigure}[b]{0.49\textwidth}
        \includegraphics[width=\textwidth]{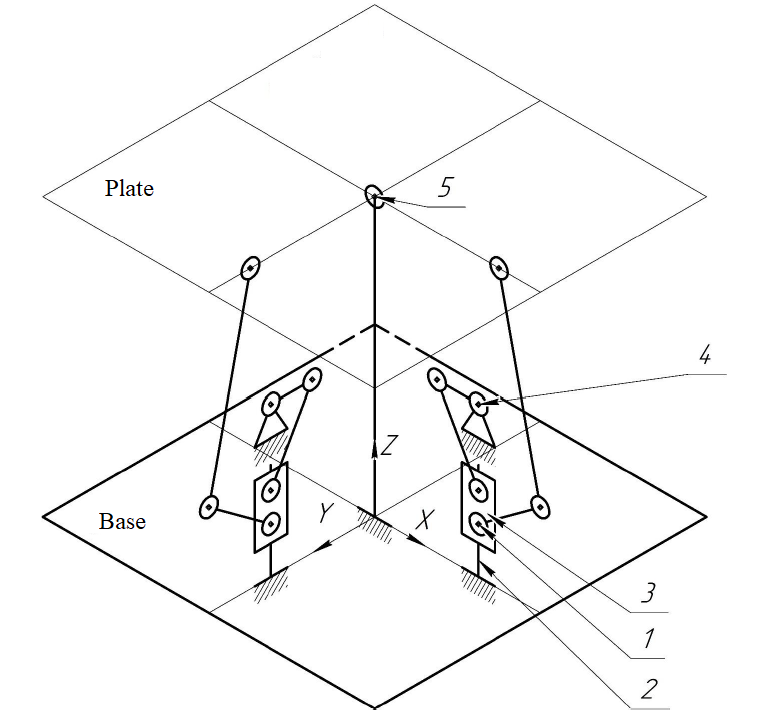}
        \caption{Ball-and-plate setup mechanical structure}
        \label{mech}
    \end{subfigure}
    \begin{subfigure}[b]{0.49\textwidth}
        \includegraphics[width=\textwidth]{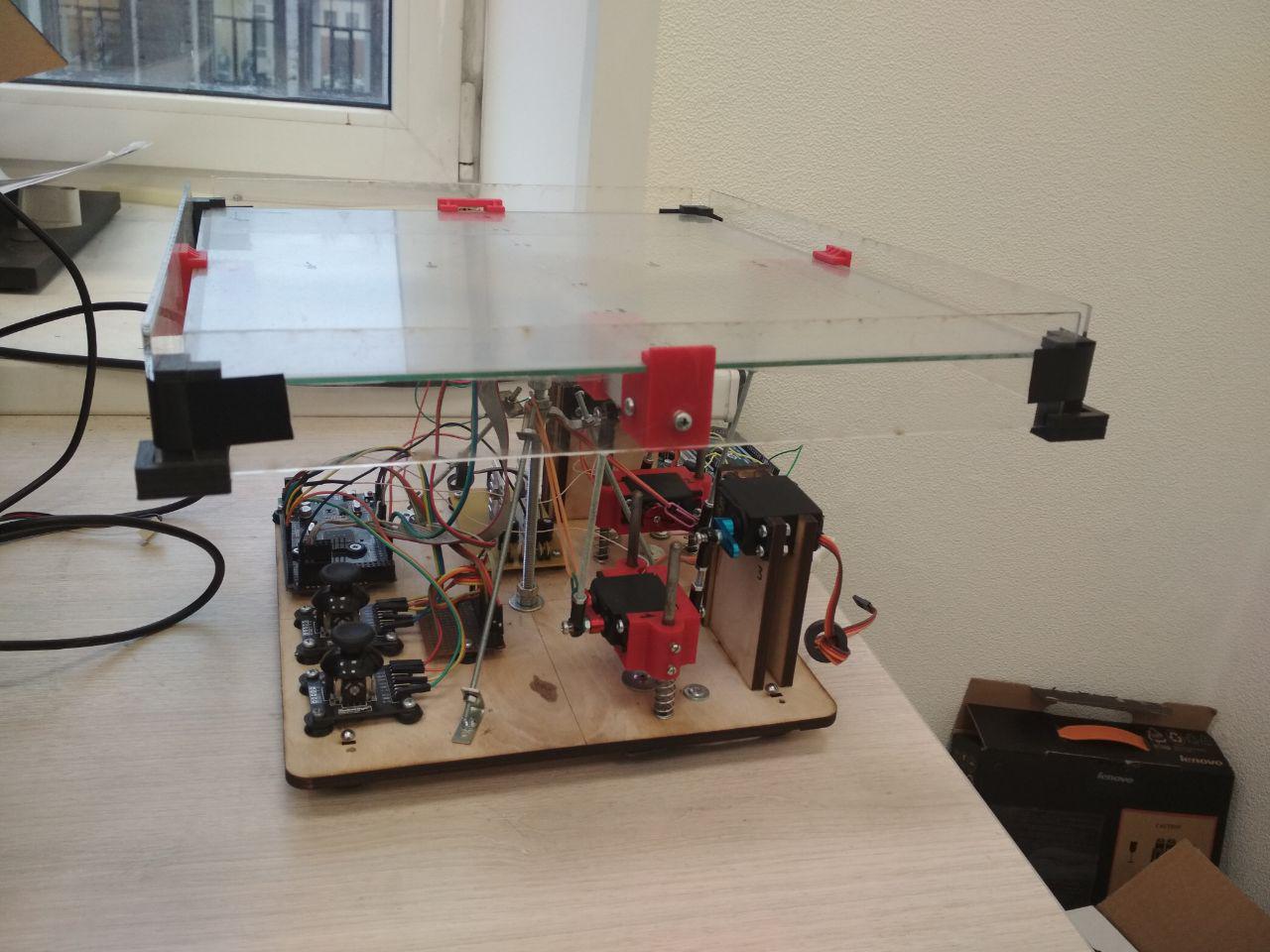}
        \caption{Ball-and-plate lab setup}
        \label{photo}
    \end{subfigure}
    \caption{Ball and Plate}\label{scheme}
\end{figure*}
\section{Case-Study Results}

In this chapter we analyse how the proposed output controller and disturbance frequency estimation algorithm working in a loop can be applied for the mechatronic setup.

Here we consider a parallel kinematics Ball-and-Plate robotic platform as a plant (see Fig.~\ref{photo}). The goal is to stabilize a steel ball in user-defined coordinates on the square plate under input harmonic disturbances by applying voltages to the servo drives controlling the inclination of a plate, while kinematic and dynamic parameters of the system are unknown. 

The ball and beam system can well be approximated by two linear decoupled systems. Therefore, this system with two inputs and two outputs can be treated as two decoupled SISO systems, therefore the proposed control approach can be implemented. 

\subsection{Setup description}

The mechatronic setup is built using four servo drives with encoders for position feedback (2 for platform inclination and 2 to generate disturbances), while a resistive touch sensor placed on top of a platform is used for ball position measurements. 

Platform itself has two degrees of freedom. Fig.~\ref{mech} illustrates the kinematic scheme of the Ball-and-plate setup. Where, 1 is control servo, 2 is binding runner for 3, 3 is binding for servo, 4 is disturbance servo, 5 is Hooke's joint. 

\subsection{Mathematical model}
Lets consider ball motion in OXZ plane (see Fig.~\ref{sketch}).

\begin{figure}[tb]
    \centering
    \includegraphics[width=\columnwidth]{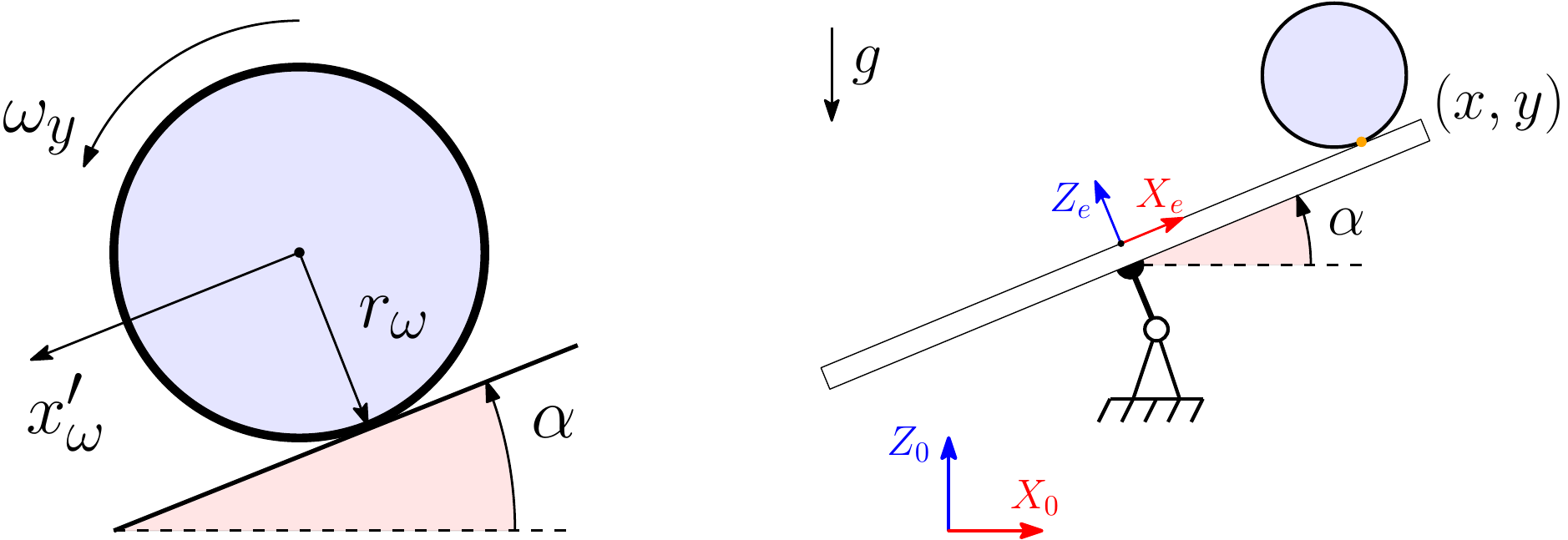}
    \caption{Sketch of a ball on a plate pictured in OXZ plane}
    \label{sketch}
\end{figure}

Here we derive the equations of motion for experimental setup under the following assumptions:
\begin{itemize}
	\item There is no slipping for ball.
	\item The ball is completely symmetric and homogeneous.
	\item Friction forces are neglected.
	\item The ball and plate are in contact all time.
\end{itemize}

By assuming the generalized coordinates of system to be $x_b$ and $y_b$ for position of the ball in each direction and $\alpha$ and $\beta$ the inclinations of the plate, i.e. $q = [x_b \, y_b \, \alpha \, \beta]^T$.

The Euler-Lagrange equations describing system dynamics are as follows \cite{Dobriborsci2018533, Dobriborsci2018655,Dobriborsci2017}:
\begin{equation}
\label{total}
\frac{d}{dt} \frac{\partial T}{\partial \dot{q}_i} - \frac{\partial T}{\partial q_i} + \frac{\partial V}{\partial q_i} = \tau_i,
\end{equation}
where $q_i$ and $\tau_i$ stand for $i$-th generalized coordinate and force respectively, $T$ is kinetic energy of the system, $V$ is potential energy of the system. 

In accordance to previous works we obtain equation of motion \cite{Dobriborsci2018533, Dobriborsci2018655,Dobriborsci2017}
\begin{subequations}
\begin{equation}
\label{eq1a}
\bigg(m_b+\frac{I_b}{r_b^2}\bigg)\ddot{x_b}-m\bigg(x_b\dot{\alpha}^2+y_b\dot{\alpha}\dot{\beta}\bigg)+m_bgsin\alpha =0
\end{equation}
\begin{equation}
\label{eq1b}
\bigg(m_b+\frac{I_b}{r_b^2}\bigg)\ddot{y_b}-m\bigg(y_b\dot{\beta}^2+x_b\dot{\alpha}\dot{\beta}\bigg)+m_bgsin\beta =0
\end{equation}
\end{subequations}

\subsubsection{Linearization}

We consider the angles of servo arms $Q_x$  and $Q_y$ as inputs, while the ball coordinates $x$ and $y$ are considered as the output. 

The relations between inclination angles of the plate $\alpha, \beta$ and $Q_x, Q_y$ are the following
\begin{subequations}
\begin{equation}
\label{eq2a}
 \alpha=\frac{d_a}{L}Q_x 
\end{equation}
\begin{equation}
\label{eq2b}
 \beta=\frac{d_a}{L}Q_y
\end{equation}
\end{subequations}
where $d_a$ is the length of the servo arm and $L$ is side length of the plate. 

In the case of a slow rate of change for the plate angles, equations \eqref{eq1a}, \eqref{eq1b} can be linearized
\begin{subequations}
\begin{equation}
\label{eq3a}
\bigg(m_b+\frac{I_b}{r_b^2}\bigg)\ddot{x}-\frac{2m_bgd}{L}Q_x=0
\end{equation}
\begin{equation}
\label{eq3b}
\bigg(m_b+\frac{I_b}{r_b^2}\bigg)\ddot{y}-\frac{2m_bgd}{L}Q_y=0
\end{equation}
\end{subequations}

The equations \eqref{eq3a}, \eqref{eq3b} are equivalent because of the symmetry of the plate. 

Taking into account servo drive dynamics, which is assumed to be captured by the 1st order aperiodic link transfer function, we can derive the following relations for the numerator and denominator of the system transfer functions for $x$ and $y$ control channels

\begin{subequations}
\begin{equation}
\label{eq4a}
b(p)=2m_bgdr_b^2K_m,
\end{equation}
\begin{equation}
\label{eq4b}
\quad\quad\quad a(p)=L(m_br_b^2T_mp^3+I_bp^2),
\end{equation}
\end{subequations}
where $K_m$ and $T_m$ are servo drives gain and time constant respectively.

Parameters of the system are presented in Table~\ref{table}.
\begin{table}[tb]
\centering
\caption{ Parameters of the Ball-and-Plate system}
\begin{tabular}{|l|l|l|l|}
\hline
$m_b$ & 0.05 $kg$    & $L$   & 0.11 $m$                          \\ \hline
$g$   & 9.81 $m/s^2$ & $K_m$ & 0.25                              \\ \hline
$d$   & 0.02 $m$     & $T_m$ & 0.018 $s$                         \\ \hline
$r_b$ & 0.0125 $m$   & $I_b$  & 3.13 $\cdot 10^{-5} kg \cdot m^2$ \\ \hline
\end{tabular}
\label{table}
\end{table}

\subsection{Simulation results}

We will verify the proposed approach for two cases. At first, consider system performance assuming that the disturbance signal is directly measurable. We demonstrate results for two harmonic disturbance signals with different parameters
\begin{equation}
\delta_1(t)=3sin(1.2t+\frac{\pi}{2}),
\label{distur1_6_1}
\end{equation}
and
\begin{equation}
\delta_2(t)=3sin(4t+\frac{\pi}{2}).
\label{distur2_6_1}
\end{equation}

\begin{figure*}[ht]
    \centering
    \begin{subfigure}[b]{0.3\textwidth}
        \includegraphics[width=\textwidth]{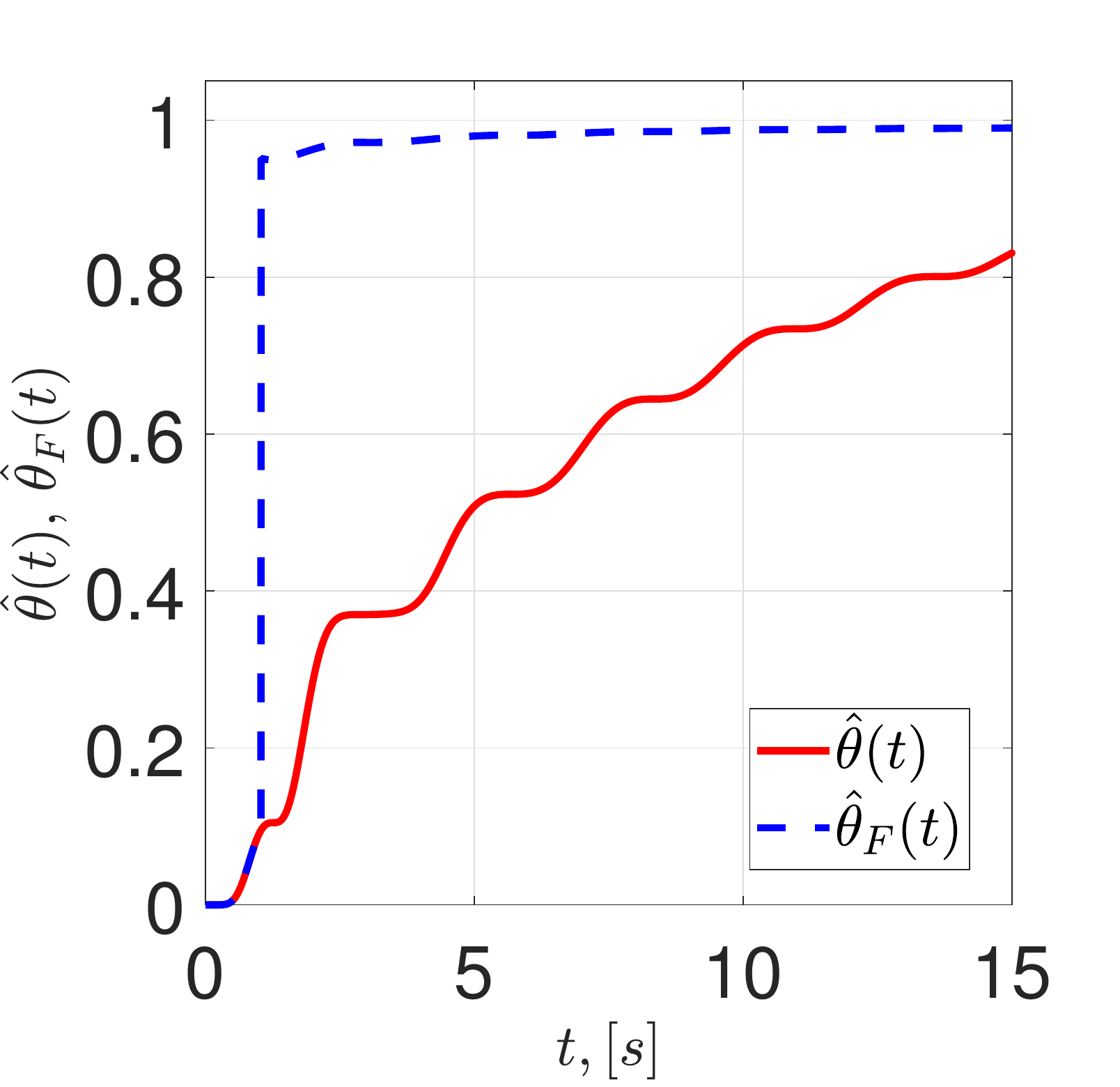}
        \caption{Standard gradient descent method and finite-time modification, for $\omega = 1.2 \frac{rad}{s}$, $K=0.5$}
    \end{subfigure}
    \begin{subfigure}[b]{0.3\textwidth}
        \includegraphics[width=\textwidth]{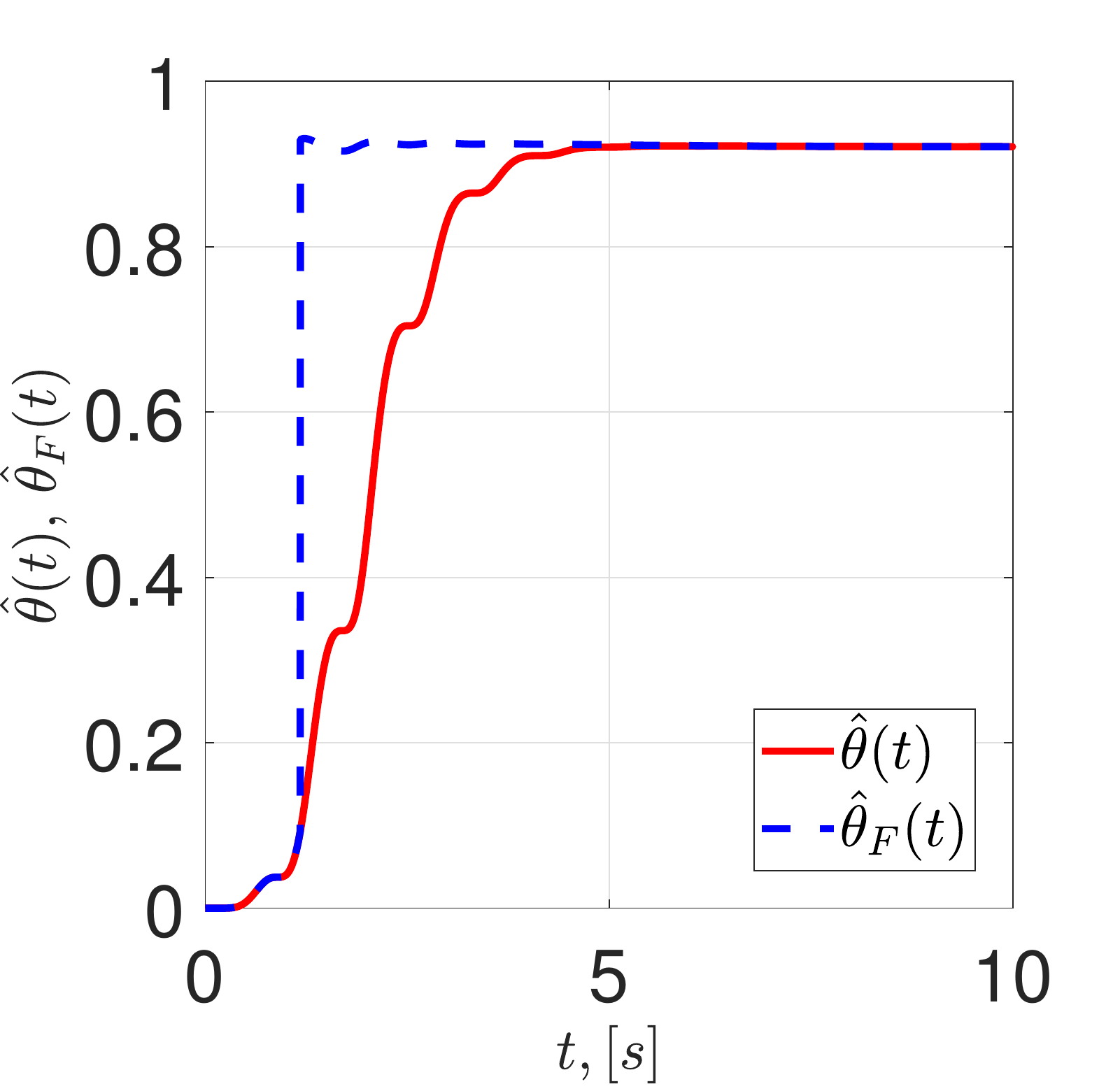}
        \caption{Standard gradient descent method and finite-time modification, $\omega = 4 \frac{rad}{s}$, $K=0.9$}
    \end{subfigure}
    \begin{subfigure}[b]{0.3\textwidth}
        \includegraphics[width=\textwidth]{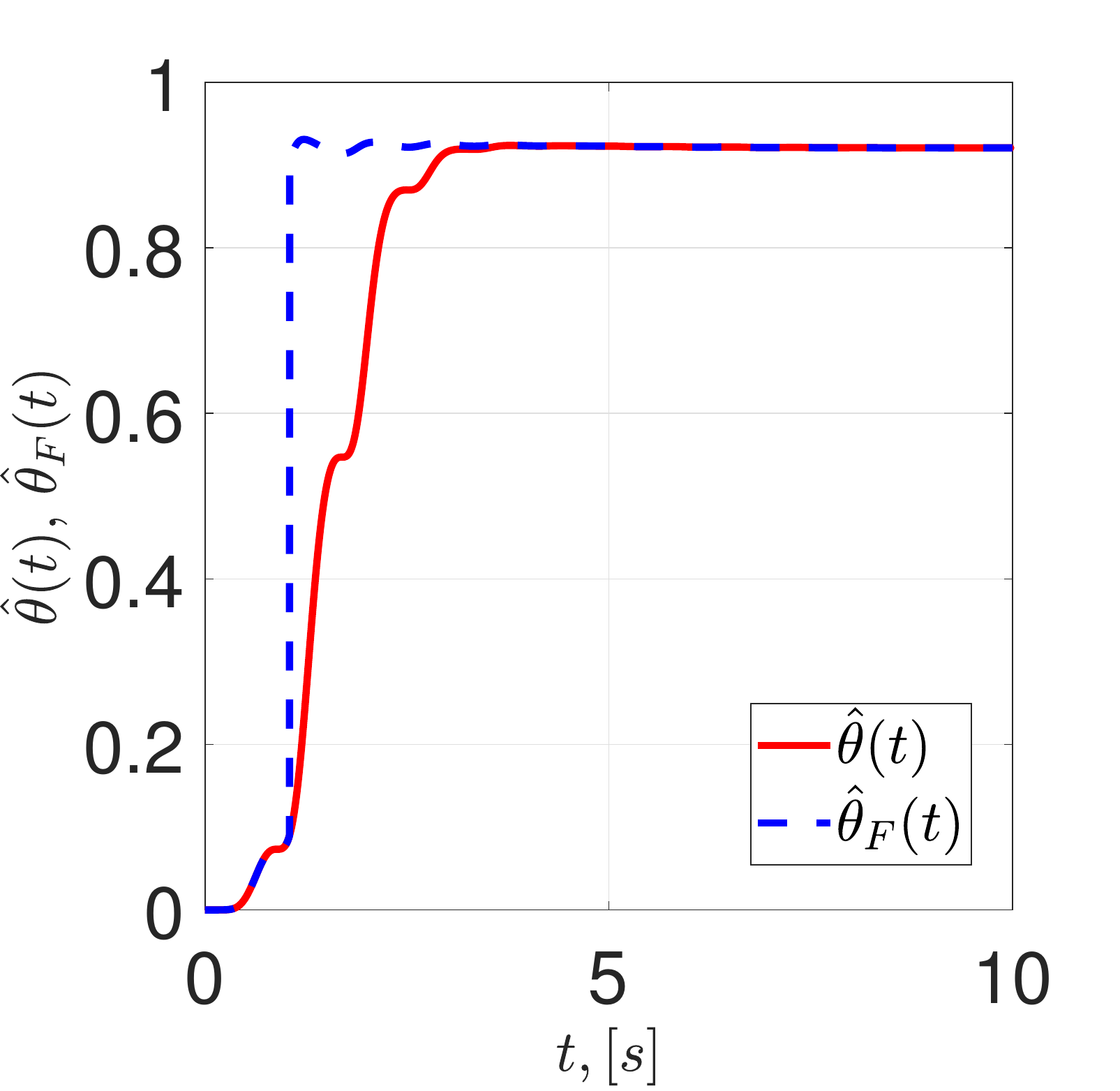}
        \caption{Standard gradient descent method and finite-time modification, where $\omega = 4 \frac{rad}{s}$, $K=1.8$}
    \end{subfigure}
    \caption{Transients for disturbance parameter estimation}
    \label{output_freq1_2_K0_5}
\end{figure*}

Transients for the signals of the closed-loop system  are presented in Fig.\ref{output_freq1_2_K0_5}, which illustrate total rejection of the input harmonic disturbance in case we use finite-time algorithm modification and its comparison to standard gradient-descent method. The increase of the gain coefficient parameter $K$ leads to the more accurate convergence of the estimates. The gradient-descent method provides the convergence of the parameters in more than 30 seconds with frequency $\omega=1.2$, $K=0.5$, $\tau=0.1$ without modification and in about 8 seconds with $K=3.8$ and with finite-time modification, whereas with frequency $\omega=4$, $K=0.9$, $\tau=0.1$ in 6 seconds without modification and in 3 seconds with $K=1.8$ with finite-time modification.

Now, consider the more realistic case, when external disturbance is not directly measurable, and we estimate its parameters by measuring system output only, while plant parameters are a priori unknown. 

For the dynamical model of the ball-and-plate setup that we use (relative degree $\rho=3$), the proposed output controller with tuned parameters can be described as
\begin{equation}
\label{controller_bp}
u(t)=-\kappa\frac{\alpha(p)(p+1)^2}{p(p^2+\bar{\omega})}\xi_1,
\end{equation}
\begin{equation}
\left\{
	\begin{array}{l}
	\dot{\xi}_1(t)=\sigma\xi_2(t),  \\
	\\ \dot{\xi}_{2}(t)=\sigma(-k_1\xi_1(t)-k_{2}\xi_{2}(t)+k_1y(t)), \\
	\\
	\end{array}
	\right.
\end{equation}
where $\kappa = 1.2, \sigma = 35, k_1 = 2, k_2 = 5$, $\alpha(p) = p^2 + 3p + 1$.

Again, we demonstrate results for two harmonic disturbance signals with different parameters \eqref{distur1_6_1} and \eqref{distur2_6_1}.
 
\begin{figure*}[ht]
    \centering
    \begin{subfigure}[b]{0.49\textwidth}
        \includegraphics[width=\textwidth]{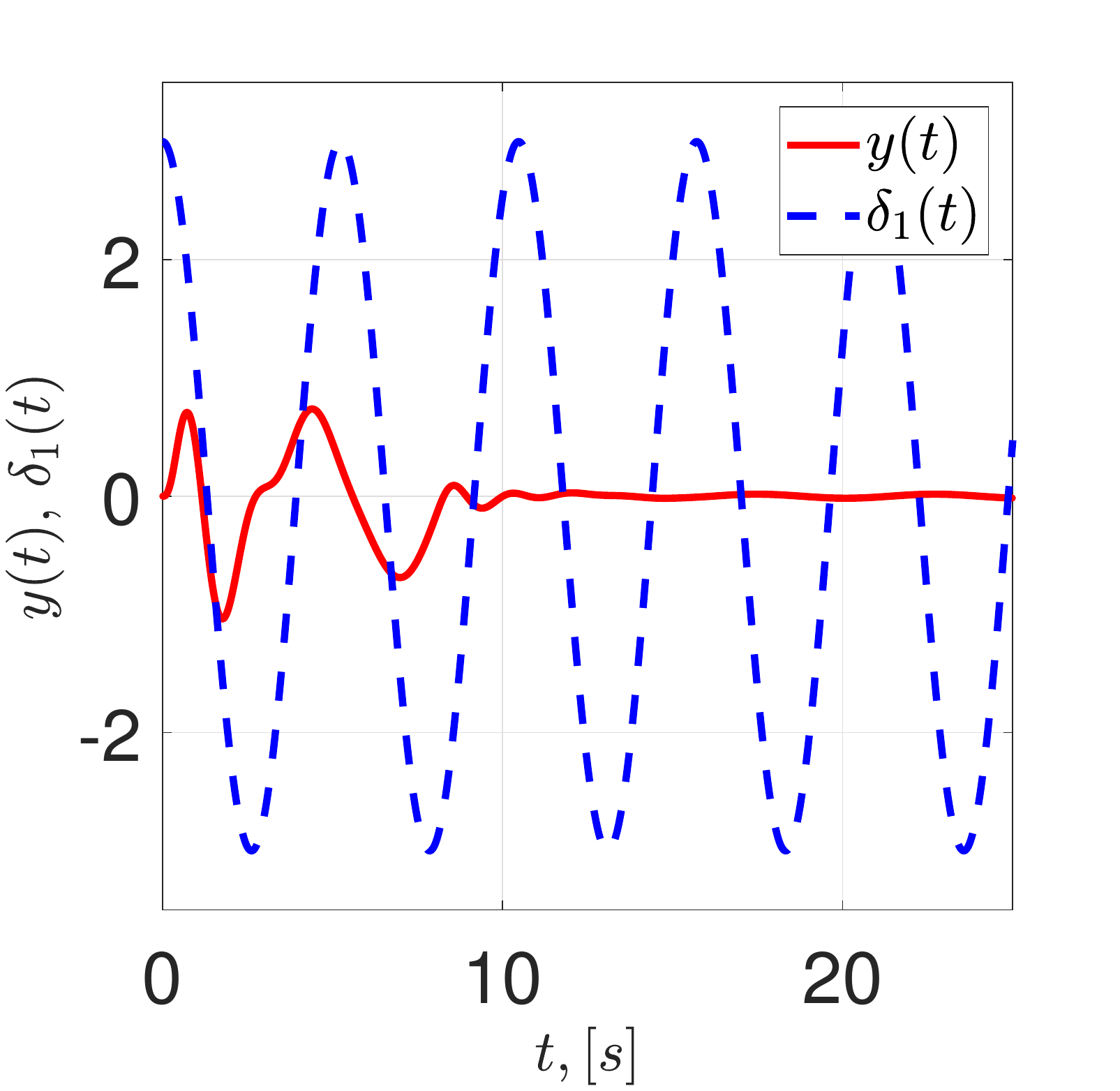}
        \caption{Output and the disturbance signals with finite-time}
    \end{subfigure}
    \begin{subfigure}[b]{0.49\textwidth}
        \includegraphics[width=\textwidth]{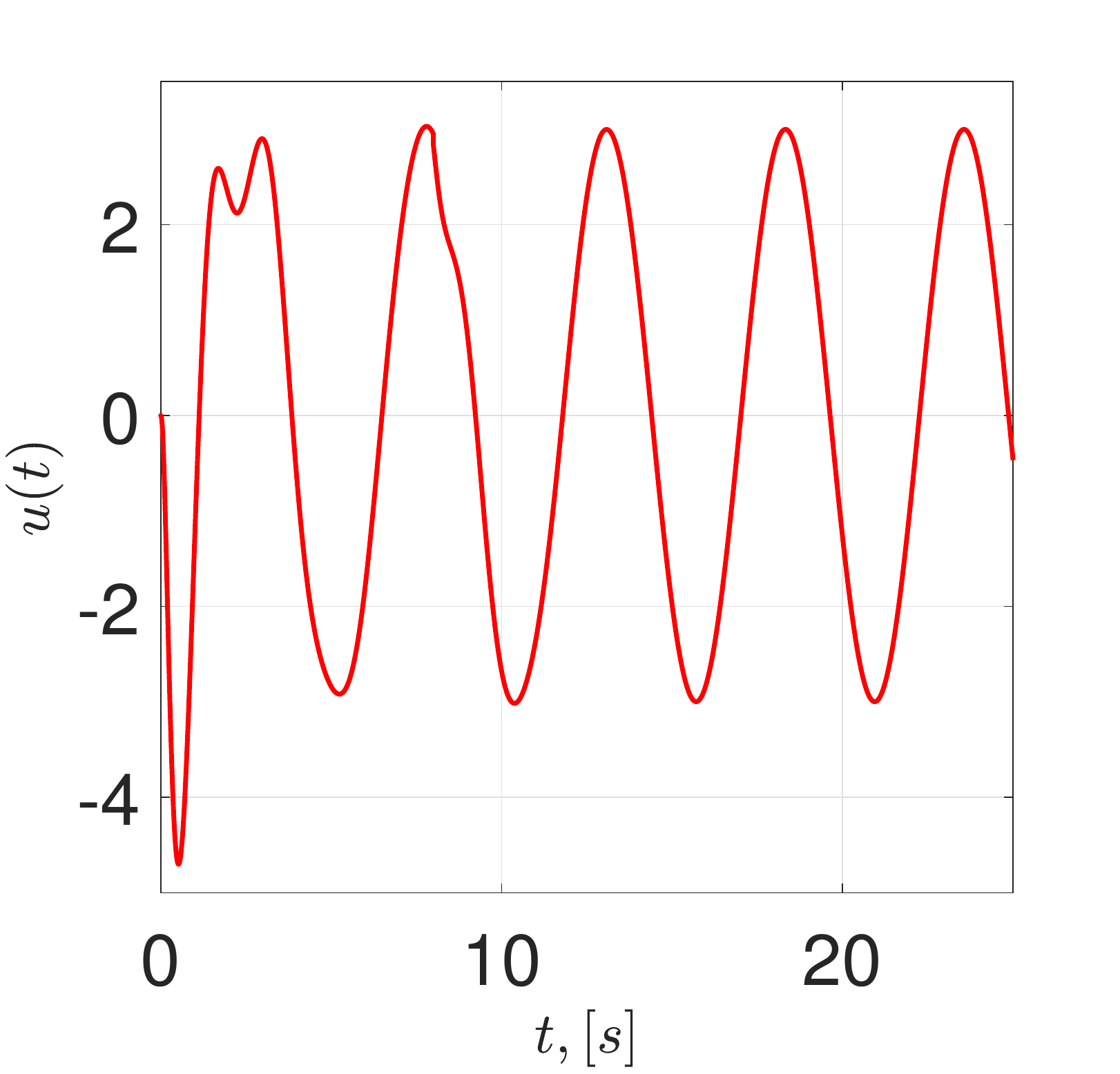}
        \caption{Control signal with finite-time modification}
    \end{subfigure}
    \caption{For $\omega = 1.2 \frac{rad}{s}$, $K=7.1$}\label{output}
\end{figure*}

\begin{figure*}[ht]
    \centering
    \begin{subfigure}[b]{0.49\textwidth}
        \includegraphics[width=\textwidth]{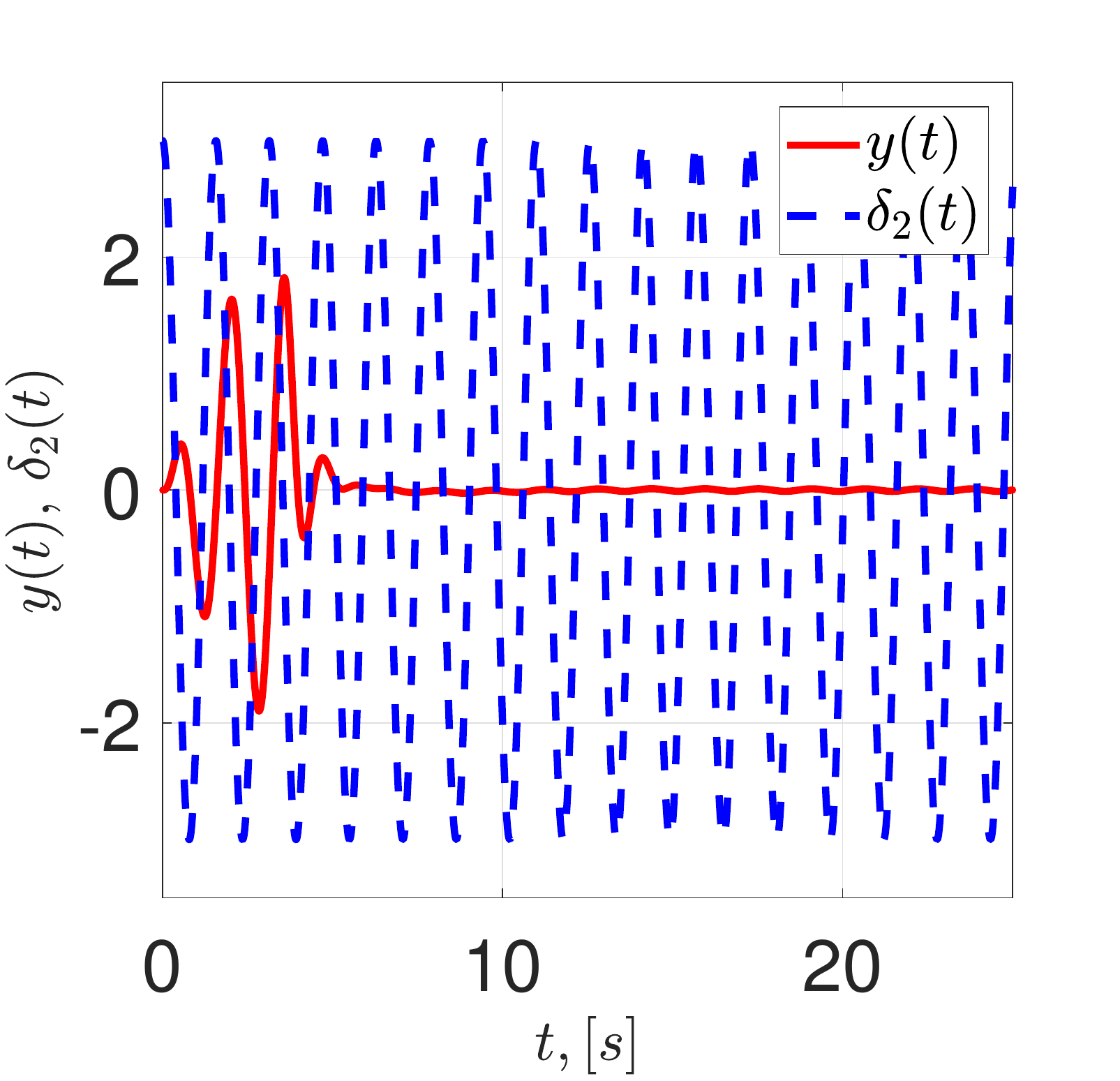}
        \caption{Output and the disturbance signals with finite-time}
    \end{subfigure}
    \begin{subfigure}[b]{0.49\textwidth}
        \includegraphics[width=\textwidth]{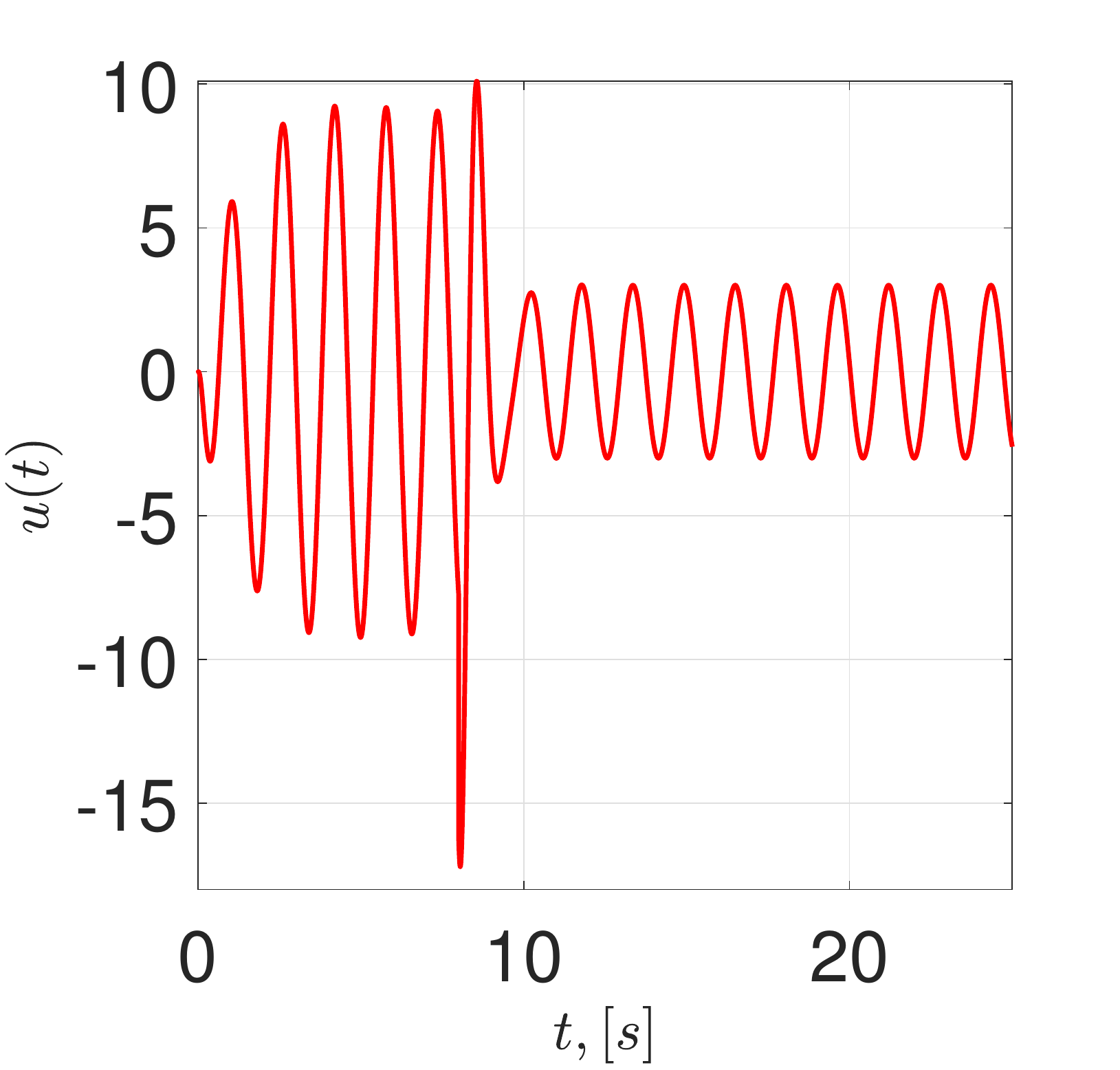}
        \caption{Control signal with finite-time modification}
    \end{subfigure}
    \caption{For $\omega = 4 \frac{rad}{s}$, $K=2.8$}\label{output2}
\end{figure*}

The simulation results of the disturbed plant behavior are presented in Fig.\ref{output2}, where identification of the frequency can be observed. The finite-time algorithm modification proves its efficiency and provides faster convergence than gradient-descent method. The transient time constitutes 8 seconds with frequency $\omega=1.2$, $K=7.1$, $\tau=0.1$  and 5 seconds with $\omega=4$, $K=2.8$, $\tau=0.1$.

\section{Conclusion}

This work presented a modification of the output adaptive control algorithm based on the ”consecutive compensator” method, where the unknown input harmonic disturbance rejection is organized via the finite-time disturbance parameters estimation algorithm. 

The proposed controller remains a simple structure, but guarantees better closed-loop system performance because of the finite-time convergence of the disturbance's parameters estimates. At the same time, this approach simplifies the switching rule used for parameters' estimates substitution to the feedback controller. 

Possible directions for future work include extension of the obtained results for the case when an input disturbance is approximated by the Fourier series, i.e. dealing with multi-harmonic signals, solving trajectory tracking tasks, including MIMO cases, and modifications of the finite-time algorithms for better robustness to measurement noise. 

\bibliographystyle{unsrt}  
\bibliography{template}

\end{document}